\documentclass[english,conference]{IEEEtran}
\usepackage[T1]{fontenc}
\usepackage[latin9]{inputenc}
\usepackage{cite}
\usepackage{amsmath}
\usepackage{amsthm}
\usepackage{amssymb}
\usepackage{mathtools}
\usepackage[normalem]{ulem}
\usepackage{bm}
\usepackage{xcolor}
\usepackage{tikz}
\usetikzlibrary{positioning,calc,arrows.meta,fit,backgrounds,decorations.pathreplacing}

\newif\ifjournalx

\journalxtrue

\newcommand{\ifjournal}{\ifjournalx} %

\newif\ifarxiv
\arxivtrue

\allowdisplaybreaks

\makeatletter
\theoremstyle{plain}
\newtheorem{thm}{\protect\theoremname}
\theoremstyle{plain}
\newtheorem{lem}[thm]{\protect\lemmaname}
\ifx\proof\undefined
\newenvironment{proof}[1][\protect\proofname]{\par
	\normalfont\topsep6\p@\@plus6\p@\relax
	\trivlist
	\itemindent\parindent
	\item[\hskip\labelsep\scshape #1]\ignorespaces
}{%
	\endtrivlist\@endpefalse
}
\providecommand{\proofname}{Proof}
\fi
\theoremstyle{remark}
\newtheorem{rem}[thm]{\protect\remarkname}
\theoremstyle{definition}
\newtheorem{defn}[thm]{Definition}

\theoremstyle{plain}

\makeatother

\usepackage{babel}
\providecommand{\corollaryname}{Corollary}
\providecommand{\lemmaname}{Lemma}
\providecommand{\remarkname}{Remark}
\providecommand{\theoremname}{Theorem}

\renewcommand{\mid}{\,|\,}
\newcommand{\F}{\mathbb{F}}

\newcommand{\cB}{\mathcal{B}}
\newcommand{\cC}{\mathcal{C}}
\newcommand{\cD}{\mathcal{D}}
\newcommand{\cE}{\mathcal{E}}

\newcommand{\cG}{\mathcal{G}}

\newcommand{\cI}{\mathcal{I}}

\newcommand{\cS}{\mathcal{S}}

\newcommand{\cV}{\mathcal{V}}

\newcommand{\cX}{\mathcal{X}}
\newcommand{\cY}{\mathcal{Y}}
\newcommand{\cZ}{\mathcal{Z}}

\newcommand{\bS}{\bm{S}}

\newcommand{\bX}{\bm{X}}
\newcommand{\bY}{\bm{Y}}

\newcommand{\ba}{\bm{a}}

\newcommand{\bc}{\bm{c}}

\newcommand{\bi}{\bm{i}}

\newcommand{\bs}{\bm{s}}

\newcommand{\bu}{\bm{u}}
\newcommand{\bv}{\bm{v}}
\newcommand{\bw}{\bm{w}}
\newcommand{\bx}{\bm{x}}
\newcommand{\by}{\bm{y}}
\newcommand{\bz}{\bm{z}}

\newif\ifold
\oldfalse

\renewcommand{\triangleq}{\coloneqq}

\DeclareMathOperator{\Aut}{PAut}   %
\DeclareMathOperator{\SAut}{SAut}   %
\newcommand{\sym}[1]{\mathbb{S}_{#1}}

\newcommand{\E}{\ensuremath{\mathbb{E}}}
\renewcommand{\Pr}{\ensuremath{\mathbb{P}}}

\newcommand{\RM}{\ensuremath{\mathrm{RM}}}

\definecolor{mygreen}{RGB}{ 0, 100, 0}

\newif\ifdraft
\draftfalse

\ifdraft
\newcommand{\JF}[1]{}
\newcommand{\NK}[1]{}
\newcommand{\HP}[1]{}
\newcommand{\GR}[1]{}
\newcommand{\GPT}[1]{}
\newcommand{\TODO}[1]{}
\else
\newcommand{\JF}[1]{\textcolor{cyan!70!black}{JF: #1}}
\newcommand{\NK}[1]{\textcolor{orange!80!black}{NK: #1}}
\newcommand{\HP}[1]{\textcolor{mygreen}{HP: #1}}
\newcommand{\GR}[1]{\textcolor{purple}{GR: #1}}
\newcommand{\GPT}[1]{\textcolor{blue}{GPT: #1}}
\newcommand{\TODO}[1]{\textcolor{red!70!black}{TODO: #1}}
\fi

\begin{document}
\title{Reed--Muller Codes Achieve the Symmetric Capacity on Finite-State Channels}
\author{
\IEEEauthorblockN{Henry D. Pfister\IEEEauthorrefmark{1},
Navin Kashyap\IEEEauthorrefmark{2},
Jean-Francois Chamberland\IEEEauthorrefmark{3},
Galen Reeves\IEEEauthorrefmark{1}}

\IEEEauthorblockA{\IEEEauthorrefmark{1}Duke University \quad
\IEEEauthorrefmark{2}IISc Bangalore \quad
\IEEEauthorrefmark{3}Texas A\&M University}
}
\maketitle

\begin{abstract}
We study reliable communication over finite-state channels (FSCs) using Reed--Muller (RM) codes. Building on recent symmetry-based analyses for memoryless channels, we show that a sequence of binary RM codes (with some random scrambling) can achieve the symmetric capacity (or uniform-input information rate) of a binary-input indecomposable FSC.

\ifjournal
Our approach has three components. First, we establish a capacity-via-symmetry theorem for doubly-transitive group codes on discrete memoryless channels (DMCs) with non-binary inputs, under some symmetry and puncturing conditions. Then, we reduce a binary-input FSC to an almost memoryless non-binary channel by grouping adjacent input bits into blocks and interleaving non-binary codes onto the channel.
Finally, we show that the interleaved non-binary codes can be constructed from a single binary RM code.
\fi

\end{abstract}


\section{Introduction}

The past decade has seen significant progress in understanding how highly-structured error-correcting codes can achieve information-theoretic limits.
In particular, much attention has been devoted to binary Reed--Muller (RM) codes~\cite{Muller-ire54,Reed-ire54}, with a goal of proving that they achieve capacity of binary memoryless channels.
This program culminated in the demonstration that RM codes are capacity-achieving on binary memoryless symmetric (BMS) channels, first for bit-error probability and then for block-error probability~\cite{Abbe-it15,Kudekar-stoc16,Kudekar-it17,Abbe-it20,hkazla-stoc21,Reeves-it23,Abbe-focs23}.

Similar symmetry-based techniques were extended to generalized RM (GRM) codes and non-binary channels in~\cite{Reeves-isit23}.
\ifjournal
A key insight of the non-binary analysis is that the proof applies broadly to code families that exhibit strong permutation and symbol symmetries (e.g., doubly-transitive group actions), together with a puncturing property which ensures that a carefully chosen set of coordinates can be removed with negligible rate change.
\fi

This work extends symmetry-based techniques to finite-state channels (FSCs), i.e., channels with memory described by a finite-state Markovian evolution~\cite{Blackwell-annmathstats58,Blackwell-annmathstats59,Gallager-1968}.
To give more context, a limited history of coding for FSCs is provided in Section~\ref{sec:history}.
For FSCs, we target the \emph{symmetric capacity} which corresponds to the information rate for i.i.d.\ uniform inputs (Definition~\ref{def:sym_capacity_fsc}).
Our main result shows that sequences of binary RM codes, augmented by a random block-scrambling and interleaving construction, can achieve the symmetric capacity on binary-input indecomposable FSCs (IFSCs)
under maximum a posteriori (MAP) decoding.

The primary challenge we face
is that the evolving state couples channel uses and breaks the independence assumptions on which existing RM proofs rely.
Classical results for indecomposable FSCs imply that, %
by interleaving many independent codes of the same rate and decoding each of them under a memoryless assumption, one can achieve the symmetric capacity of the stationary averaged channel; however, this rate is usually less than the FSC symmetric capacity.
By applying decision-feedback equalization (DFE)
techniques, one can recoup the loss by optimizing the rates of the interleaved codes~\cite{Mushkin-it89,Goldsmith-it96,Pfister-globe01,Soriaga-it07,Narayanan-aller04}.
Our approach lifts the FSC into a form that is both approximately memoryless and highly symmetric, so that the existing RM machinery can be applied at rates approaching the FSC symmetric capacity.
A key difficulty lies at the boundaries between RM symbols, since recent symbols are informative about the current channel state; naively discarding this information can degrade performance.
Controlling boundary effects is therefore central to the construction. \\[0.5mm]
We now explain how this is accomplished.

\vspace{0.5mm}
\noindent\emph{(i) From FSCs to interleaved non-binary channels.}
One approach to attenuate the memory of an FSC %
is to interleave multiple subcodes, inserting large gaps between bits %
belonging to the codeword so that interleaves behave nearly independently.
However, because nearby bits carry information about the channel state, ignoring previous bits leads to some loss in performance.
Fortunately, this issue can be mitigated by grouping adjacent inputs into blocks of bits and treating each block as a non-binary symbol, thereby minimizing the impact of the boundary effects.

Adopting this strategy, we treat the sequence of input blocks as $d$ interleaved non-binary group codes.
To understand the performance, we can distinguish a single interleave and focus on it as a representative of all others.
In this setup, we refer to blocks in the distinguished interleave as \emph{protected}, and these blocks are separated by $d-1$ blocks in the input stream.
For FSCs with exponential mixing of the state process, this separation produces approximate independence of the boundary states across protected blocks, which in turn yields an approximately memoryless block channel.

To ensure strong symmetry of the induced block channel (and uniform-input behavior even when the state evolution depends on the input), we apply independent random scrambling within each protected block and reveal the scrambling to the decoder.
The resulting induced channel is compatible with the memoryless capacity-via-symmetry theorem.

\vspace{1mm}
\emph{(ii) Capacity-via-symmetry for non-binary codes.}
We distill from~\cite{Reeves-isit23} a general theorem (Theorem~\ref{thm:capacity_symmetry_subcode_app}) for symmetric memoryless channels with a non-binary input alphabet $\cX$ that forms a group.
This result states that a sequence of doubly-transitive group codes over $\cX$ achieves the symmetric capacity of the channel, with respect to symbol-error probability under symbol-MAP decoding, provided that: (a) the code satisfies suitable symmetry conditions; (b) a mild non-degeneracy condition on the channel holds; and (c) a \emph{two-look puncturing} condition is satisfied.

\vspace{1mm}
\emph{(iii) Binary RM codes are sufficient.}
To use binary RM codes for this problem, we first use Theorem~\ref{thm:capacity_symmetry_subcode_app} to prove Theorem~\ref{thm:capacity_symmetry_subcode} in the memoryless setting.
Then, we rely on the ability of binary RM codes to support \emph{implicit interleaving and puncturing} with negligible rate loss.
In particular, if a single binary RM codeword is transmitted over an FSC, then the set of bits in the protected symbols of a single interleave also form a codeword, albeit from a shorter RM code of slightly higher rate.
We use the idea of implicit interleaving to bound the performance of a single long code by treating it as multiple interleaved shorter codes.
Moreover, for long enough codes, one can simultaneously make the rate difference arbitrarily small and the decimation factor $d$ arbitrarily large.
The interleaving property follows from the large automorphism group of RM codes, while the puncturing property follows from their nested construction.
\ifarxiv
We note that these properties have been exploited previously in proofs that RM codes achieve capacity~\cite{Kudekar-it17,Reeves-it23,Abbe-focs23}.
\fi

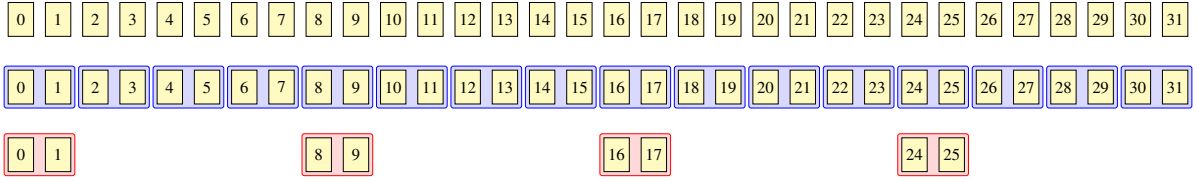
\begin{figure*}[t]
\centering
\scalebox{0.9}{%
\begin{tikzpicture}[
  font=\small,
  bit/.style={fill=yellow!30!white, draw, minimum width=3.75mm, minimum height=5mm, inner sep=0pt, align=center},
  blk/.style={draw, rounded corners=1pt, inner sep=1.5pt}
]

\foreach \i in {0,...,31} {
  \node[bit] (s\i) at ({-4+0.55*\i},2) {\scriptsize \i};
}

\foreach \i in {0,...,31} {
  \node[bit] (t\i) at ({-4+0.55*\i},1) {\scriptsize \i};
}

\foreach \i in {0,1,8,9,16,17,24,25} {
  \node[bit] (r\i) at ({-4+0.55*\i},0) {\scriptsize \i};
}

\begin{scope}[on background layer]
  \foreach \b in {0,4,...,12} {
    \pgfmathtruncatemacro{\a}{2*\b}
    \pgfmathtruncatemacro{\c}{2*\b+1}
    \node[blk, fit=(r\a)(r\c), red, fill=red!15] {};
  }
\end{scope}

\begin{scope}[on background layer]
  \foreach \b in {0,4,...,12} {
    \pgfmathtruncatemacro{\a}{2*\b}
    \pgfmathtruncatemacro{\c}{2*\b+1}
    \node[blk, fit=(t\a)(t\c), blue, fill=blue!15] {};
  }
\end{scope}

\begin{scope}[on background layer]
  \foreach \b in {1,2,3,5,6,7,9,10,11,13,14,15} {
    \pgfmathtruncatemacro{\a}{2*\b}
    \pgfmathtruncatemacro{\c}{2*\b+1}
    \node[blk, blue, fill=blue!15, fit=(t\a)(t\c)] {};
  }
\end{scope}
\end{tikzpicture}}
\caption{Bits to blocks to interleaves for $N=32$ bits grouped into blocks of $b=2^h$ bits (e.g., with $h=1$) and $d=2^g$ interleaves (e.g., with $g=2$).
The last row shows a single representative interleave of 4 protected blocks after decimation by $d$.}
\label{fig:block_decimate_structure}
\end{figure*}

\begin{thm}[Informal] %
Consider a binary-input indecomposable FSC that is injective with symmetric capacity $C$ bits per channel use.
For any rate $R<C$, let $\cC_m = \RM(r_m,m)$ be a sequence of binary RM codes with rates $R_m\to R$ and let $h\colon \mathbb{N} \to \mathbb{N}$ satisfy $h(m)\to \infty$ and $h(m)/m\to 0$.
Then, MAP decoding of $\cC_m$ (with random scrambling applied to blocks of $2^{h(m)}$ bits) has vanishing symbol error probability as $m\to\infty$.
\end{thm}

\noindent
This result is restated in more detail as Theorem~\ref{thm:main}.

\section{Brief History of Capacity and Coding for FSCs}
\label{sec:history}

FSCs provide a broad abstraction for channels with memory evolving on a finite state space.
They were introduced by Blackwell, Breiman, and Thomasian, who established a capacity theorem for \emph{indecomposable} FSCs, providing the first operational meaning of FSC capacity and identifying
regularity conditions under which information stability holds~\cite{Blackwell-annmathstats58,Blackwell-annmathstats59}.
Gallager's 1968 textbook extended these ideas and distilled them into a unified treatment of channels with memory~\cite{Gallager-1968}.

Starting in the late 1980s and accelerating in the 2000s, there has been sustained interest in the problems of efficiently estimating their capacity and designing coding schemes to approach it.
In 1989, Mushkin and Bar-David~\cite{Mushkin-it89} made progress on both these problems for the Gilbert--Elliott burst-noise channel~\cite{Gilbert-bell60,Elliott-bell63} which is the canonical example of a fading-type FSC.
In particular, they discussed interleaving as a technique to control correlation and decision-feedback as a tactic to recoup the rate lost by neglecting the correlation.
These ideas were later extended to more practical models for fading channels~\cite{Wang-vt95,Goldsmith-it96,Sadeghi-sp08}.

Coding for FSCs was greatly improved by the breakthrough introduction of turbo codes, turbo equalization, and iterative information processing~\cite{Berrou-icc93,Douillard-ett95}.
Many groups then turned their attention to coding for FSCs that model intersymbol interference where the state is a deterministic function of past inputs~\cite{Fan-aller99,Souvignier-com00,Kavcic-it03}.
For these channels, simulation-based estimators of the mutual information were proposed and found to be surprisingly effective~\cite{Arnold-icc01,Pfister-globe01,Arnold-it06,Pfister-it07}.
These ideas were also paired with interleaved optimized codes to approach capacity~\cite{Pfister-globe01,Pfister-it07}.
Efficient optimization of the input distribution has also been studied~\cite{Kavcic-globe01,Soriaga-isit04,Vontobel-it08,Han-it15,Wu-it22}.

A common practical coding technique for FSCs is to use a single LDPC code with turbo equalization.
One drawback of this approach is the associated increase in decoding complexity.
This issue can be mitigated by using interleaved codes together with a decision-feedback scheme~\cite{Narayanan-aller04}.
Another drawback is that, in order to approach the symmetric capacity, the LDPC code(s) must be optimized for the particular FSC.
This limitation can be overcome by using universal codes such as spatially coupled LDPC codes~\cite{Kudekar-isit11-DEC,Nguyen-icc12,Abe-isit16}.
Polar codes with successive-cancellation trellis decoding can also approach the symmetric capacity and, in some cases, even the true capacity~\cite{Wang-itw15,Sasoglu-isit16,Sasoglu-it19}.

These prior contributions provide many paths toward capacity but do not identify explicit structured code families that achieve it.
In this work, we rely on the implicit interleaving property of RM codes to show that binary RM codes can universally achieve the symmetric capacity of primitive FSCs. We note here that coding schemes involving RM codes have been previously proposed~\cite{rameshwar-it23} for input-constrained BMS channels (which are a special case of FSCs), but the coding schemes in that work are not capacity-achieving.

\subsection{Organization}

Section~\ref{sec:background} presents necessary background, including a review of FSCs and Reed--Muller codes. Section~\ref{sec:general_capacity_thm} contains a statement of the (existing) non-binary capacity-via-symmetry theorem for memoryless channels.
Section~\ref{sec:rm_fsc_decimation} develops the blocking and decimation ideas used to approach the symmetric capacity on IFSCs. 
Section~\ref{sec:conclusion} offers some conclusions.

\section{Background} \label{sec:background}

\subsection{Preliminaries}
Let $\cX$ be a finite set with $q$ elements.
Information rates of channels with input alphabet $\cX$ and codes over $\cX$ are measured in \emph{qits} (i.e., logarithm base-$q$ units).
All other logarithms are natural.
\ifarxiv
We use bold lowercase (e.g., $\bx$) for vectors and calligraphic letters (e.g., $\cX$) for sets.
We represent the binary field by $\F_2$ and, more generally, $\F_q$ denotes a finite field of size $q$ (a prime power).
For a positive integer $N$, we write
\[
[N]\triangleq\{0,1,\ldots,N-1\}.
\]
For a finite set $\cX$, we write $\sym{\cX}$ for the symmetric group of all permutations of $\cX$ and, in particular, $\sym{N}$ denotes the permutations of $[N]$.
The expectation of random variable $Z$ is denoted by $\E[Z]$, and the probability of event $E$ is $\Pr(E)$.
The total-variation distance between distributions $P$ and $Q$ on a finite set $\cZ$ is
\[
\|P-Q\|_{\mathrm{TV}} \triangleq 
\frac{1}{2} \sum_{z \in \cZ} \left| P(z) - Q(z) \right|.
\]
For a vector $\bc=(c_0,\ldots,c_{N-1})\in \cX^N$ and an integer $d\ge 1$, we define \emph{decimation by $d$} as
\begin{equation} \label{eq:decimation}
\bc^{\downarrow d} \ \triangleq\ (c_0,c_d,c_{2d},\ldots,c_{\lfloor (N-1)/d\rfloor d})\in\cX^{\lfloor (N-1)/d\rfloor+1}.
\end{equation}

\fi

\subsection{Finite-State Channels and Achievable Rates}
\label{sec:fsc}

Consider a \emph{finite-state channel} (FSC) with input alphabet $\cX$, output alphabet $\cY$, and state set $\cS$ (all finite)~\cite{Blackwell-annmathstats58,Gallager-1968}.
For a length-$N$ random input $\bX=(X_0,\ldots,X_{N-1}) \in \cX^N$, this channel outputs a random sequence $\bY = (Y_0,\ldots,Y_{N-1}) \in \cY^N$ that depends on the hidden random state sequence $\bS = (S_0,\ldots,S_N) \in \cS^{N+1}$.
The channel law is specified by the conditional probability\footnote{For notational convenience, we let $S_i$ be channel state that determines the channel from $X_i$ to $Y_i$ in contrast to~\cite{Gallager-1968} which uses $S_{i-1}$.}
\begin{align*}
P_{Y, S' \mid X,S}&(y, s' \mid x,s) \\
\ & =\
\Pr(S_{i+1}=s',\,Y_i=y\mid X_i=x,\,S_i=s),
\end{align*}
where $X_i\in\cX$ is the channel input at time $i$, $Y_i\in\cY$ is the corresponding output, $S_i\in\cS$ is the channel state for input $i$, and $S_{i+1}\in\cS$ is the subsequent state. 
Such a FSC is denoted by the tuple $(\cX,\cY,\cS,P)$.

For a known distribution on the initial state $S_0$, the channel induces a conditional joint distribution on output $\bY$ and state sequence $\bS$ given the input $\bX$.
Averaging over the state distribution gives the input-output relationship
\begin{align*}
\Pr(\bY&=\by \mid \bX=\bx)\\
&= \sum_{\bs \in\cS^{N+1}}
\Pr(S_0=s_0) \prod_{i=0}^{N-1} P_{Y, S'\mid X,S}(s_{i+1},y_i\mid x_i,s_{i}) .
\end{align*}
The input-controlled state-transition probabilities are
\[
P_{S'|X,S} (s' \mid x,s) \coloneqq \sum_{y\in \cY} P_{Y, S' \mid X,S}(y, s' \mid x,s)
\]
and the \emph{uniform state-transition probability matrix}
\[
\overline{\Phi}_{s,s'} \coloneqq \frac{1}{|\cX|} \sum_{x\in\cX} P_{S'\mid X,S}(s'\mid x,s).
\]

\begin{defn}
	An FSC is called \emph{primitive} if the uniform state-transition probability matrix $\overline{\Phi}$ is primitive.
\end{defn}

\begin{rem}
An FSC is primitive if and only if there exists an $L < \infty$ such that, for all $s,s'\in \cS$, there is an $\bx=(x_0,\ldots,x_{L-1}) \in \cX^L$ satisfying
\[ \Pr(S_{L}=s' \mid S_0= s,\,X_0 = x_0,\ldots, X_{L-1}=x_{L-1}) > 0. \]
In such a case, it follows that $\overline{\Phi}^L_{s,s'} > 0$ for all $s,s'\in \cS$.
\end{rem}

\begin{defn}
	An FSC is called \emph{indecomposable} if, for any $\epsilon >0$, there is an $L_0 < \infty$ such that, for all $L>L_0$, we have
\[
\begin{split}
\big| & \Pr(S_{L}=s_{L} \mid S_0=s_0, X_0=x_0, \ldots,X_{L-1}=x_{L-1}) \\
&- \Pr(S_{L}=s_{L} \mid S_0=s_0',\,X_0 = x_0,\ldots,X_{L-1}=x_{L-1}) \big| \!<\! \epsilon
\end{split}
\]
for all states $s_0,s_0',s_{L} \in\cS$ and all input vectors $\bx=(x_0,\ldots,x_{L-1}) \in \cX^L$.
This condition implies that the channel forgets its initial state regardless of the input sequence.
\end{defn}

\begin{defn}
\label{def:sym_capacity_fsc}
Let $\bX^n$ be an i.i.d.\ uniform input process (i.e., $X_i\sim \mathrm{Unif}(\cX)$) and let $\bY^n$ be the resulting FSC output process for initial state $S_0 = s_0$.
The \emph{symmetric capacity} (also called the symmetric information rate) is defined by
\begin{equation} \label{eq:SymmetricCapacity}
C_{\mathrm{sym}}
\ \triangleq\
\lim_{n\to\infty}\frac{1}{n} I(\bX^n;\bY^n \mid S_0=s_0), 
\end{equation}
when the limit exists for all $s_0 \in \cS$ and does not depend on $s_0$. 
It is measured in logarithm base-$|\cX|$ units per channel use.
\end{defn}

For indecomposable FSCs, the limit in \eqref{eq:SymmetricCapacity} exists and does not depend on the initial state $s_0 \in \cS$
\cite{Blackwell-annmathstats58,Gallager-1968}.

\begin{lem} \label{lem:ind2prim}
If an FSC is indecomposable, then the Markov chain defined by $\overline{\Phi}$ has a unique stationary distribution $\pi$ that is supported on the unique recurrent class.
Thus, the chain formed by removing all transient states is a primitive FSC with the same symmetric capacity as the original FSC.
\end{lem}
\ifarxiv
 \begin{proof}
   See Appendix~\ref{sec:proof_lem_ind2prim}.
 \end{proof}
\fi

\begin{defn} A \emph{fading-type} FSC is a special class of FSC in which the state evolution is input-independent, i.e., there exists a state transition matrix $\Phi\in\mathbb{R}^{|\cS|\times|\cS|}$ such that
\[
P_{Y, S'\mid X,S}(y, s'\mid x,s)
\ =\
W_s(y\mid x)\,\Phi_{s,s'},
\]
where $W_s(\cdot\mid\cdot)$ is the input-output relation for state $s$.
\end{defn}
\begin{rem}
For a fading-type FSC, the state process $(S_i)_{i\ge 0}$ is a Markov chain with transition matrix $\Phi$. If the FSC is primitive, then $\Phi$ is primitive and the Markov chain has a unique stationary distribution $\pi$ and mixes exponentially fast.
\end{rem}

\begin{defn}
\label{def:CodeRate}
A \emph{code} of length $N$ over an alphabet $\cX$ is a subset $\cC\subseteq \cX^N$.
For such a code with $|\cX|=q$, we denote the rate by
\[
R(\cC)\ \triangleq\ \frac{1}{N}\log_{q}|\cC|,
\]
measured in qits per channel use.
If $\cX=\F_q$ and $\cC$ is an $\F_q$-linear subspace of $\F_q^N$ of dimension $K$, then $\cC$ is an $[N,K]$ linear code with $|\cC|=q^K$ and $R(\cC)=K/N$.
\end{defn}

\begin{defn}
The \emph{permutation automorphism group} $\cG$ of a code $\cC$ is defined by
\[ \cG \coloneqq \left\{ \pi \in \sym{N} \mid \forall c \in \cC, \pi c \in \cC \right\}. \] 
\end{defn}

\begin{defn}
Given a subset $J \subseteq [N]$, the \emph{punctured code} $\cC|_{J}\subseteq \cX^{|J|}$ is the projection of $\cC$ onto the coordinates in $J$.
For $J_d=\{0,d,2d,\dots\}$, note that $\bc \in \cC \; \Rightarrow \; \bc^{\downarrow d} \in \cC|_{J_d}$.
\end{defn}

\ifjournal
\begin{rem}[Achievable information rates]
Standard random-coding arguments for FSCs imply that, for any $R<C_{\mathrm{sym}}(W)$, there exist sequences of codes with asymptotic rate $R$ and vanishing block error probability.
Furthermore, when $|\cX| = q$ is a prime power, the same achievability holds for the ensemble of random linear coset codes over $\F_q$.
Indeed, the random-coding argument carries over because pairwise codeword differences are uniformly distributed (e.g., see \cite[Ch.~7]{Gallager-1968}).
\end{rem}
\fi

\subsection{Reed--Muller Codes}

In 1954, Reed--Muller (RM) codes were introduced by Muller~\cite{Muller-ire54} and a low-complexity decoder was described by Reed~\cite{Reed-ire54}. They are binary linear codes formed by evaluating binary multivariate polynomials of restricted degree over all possible evaluation points.
The following well-known properties of RM codes can be found in many papers~\cite{Abbe-now23,Reeves-it23}.

\begin{defn}
Let $\mathcal{F}(r,m)$ denote the set of $m$-variate multilinear functions over $\mathbb{F}_{2}$ with degree at most $r$. Thus, each $f\in\mathcal{F}(r,m)$ is a mapping $f\colon\mathbb{F}_{2}^{m}\to\mathbb{F}_{2}$ and we can write
\begin{align*}
\mathcal{F}(r,m)
&\triangleq \left\{ f(\bv) \;= \sum_{\bi\in\{0,1\}^{m}:|\bi|\leq r} \!\! a_{\bi}\,\bv^{\bi}\,\Bigg|\,a_{\bi}\in\mathbb{F}_{2}\right\} ,
\end{align*}
where $\bv^{\bi} = \prod_{j=0}^{m-1}v_{j}^{i_{j}}$ and $|\bi| \triangleq \sum_{j=0}^{m-1} i_{j}$.
For evaluations with $\bv \in \mathbb{F}_2^m$, the multilinear functions are equivalent to multivariate polynomials because $v^{i}=v$ for $v\in\mathbb{F}_{2}$ and $i\geq 1$.
\end{defn}

\begin{defn} \label{def:theta}
For $N=2^{m}$, let $\theta_{m}\colon[N]\to\mathbb{F}_{2}^{m}$ be the bijective map that takes integer $\ell \in [N]$ to its binary expansion $\ba = (a_0,\ldots,a_{m-1})\in\mathbb{F}_{2}^{m}$, where $a_{0}$ is the LSB and $a_{m-1}$ is the MSB.
For $m=5$ and $N=32$, the bijection yields $\theta_{5}(19)=(1,1,0,0,1)$.
Let $f$ be an $m$-variate polynomial over $\mathbb{F}_2$. Using this one-to-one correspondence, we take the evaluation of $f$ to be $\bc$ where $c_{\ell}=f\big(\theta_m (\ell)\big)$ for all $\ell \in [N]$.
\end{defn}

\begin{defn} \label{def:rm}
For $N\!=\!2^{m}$, let $\mathrm{RM}(r,m)\!\subseteq\!\mathbb{F}_{2}^{N}$ be the length-$N$ Reed--Muller code defined by evaluating all functions in $\mathcal{F}(r,m)$ on every point in $\mathbb{F}_{2}^{m}$. Equivalently, we can write
\[ \scalebox{0.95}{$\mathrm{RM}(r,m) \!\triangleq\!
\big\{ \bc\in\mathbb{F}_{2}^{N}\,\big|\,f\in\mathcal{F}(r,m),\,\! c_{\ell}=f\big(\theta_m (\ell)\big),\, \ell\in[N] \big\}.$} \]
\end{defn}

\begin{lem} \label{lem:rm_affine_sym}
The permutation-automorphism group of $\cC = \RM(r,m)$ is doubly transitive and generated by all invertible affine transformations of the codeword index space $\F_2^m$.
\end{lem}

\begin{lem}
\label{lem:dim_rm}
$\mathrm{RM}(r,m)$ has dimension $\binom{m}{\leq r} \triangleq \sum_{i=0}^{r}\binom{m}{i}$.
\end{lem}

\begin{lem}[$\!\!$\cite{Reeves-it23}] \label{lem:rm_rate_diff}
For $\cC_k = \RM(r,m+k)$, the rates satisfy $R(\cC_{k}) \geq R(\cC_0) - \frac{k}{2\sqrt{m}}$.
Thus, the rate decrease due to increasing $k$ is negligible if $k$ grows $o(\sqrt{m})$.
\end{lem}

Sequences of $\RM$ codes are known to achieve the capacity $C$ of a binary memoryless symmetric (BMS) channel~\cite{Reeves-it23,Reeves-isit23,Abbe-focs23,Pfister-arxiv25a}.
For asymmetric binary memoryless channels, if one uses a random coset of each code, then standard arguments~\cite{Hou-it03} show that the modified sequence achieves the symmetric capacity $C_\mathrm{sym}$ described in Definition~\ref{def:sym_capacity_fsc}.

A key characteristic of $\RM$ codes is that, if one evaluates the message polynomials on an $n$-dimensional affine subspace of $\F_2^m$, then one obtains the code $\RM(r,n)$~\cite{Reeves-it23}.
Choosing the subspace whose first $g$ coordinates (i.e., LSBs) are fixed to $0$ defines a mapping from $\bc \in \RM(r,m)$ to $\bc' \in \RM(r,m-g)$ where $\bc'_{\ell} = \bc_{\ell 2^g}$, $\ell\in[2^{m-g}]$.
This implies that $\bc'$ is the result of decimating $\bc$ by a factor of $d=2^g$, an action denoted by $\bc' = \bc^{\downarrow d}$ (see~\eqref{eq:decimation}).
In addition, the resulting increase in code rate is negligible as $m\to \infty$ if $g=o(\sqrt{m})$~\cite{Reeves-it23}.

\section{Capacity via Symmetry}
\label{sec:general_capacity_thm}

This section outlines necessary background for the non-binary capacity-via-symmetry theorem on memoryless channels and then provides a formal statement.
In particular, the theorem describes a capacity-achieving criterion for the symmetric capacity on certain memoryless channels with input alphabet $\cX$ based on two ingredients:
(i) code/channel symmetries that enable the overlap-matrix analysis of~\cite{Reeves-isit23}; and
(ii) a \emph{puncturing condition} asserting that one can puncture a carefully chosen set of coordinates and only increase the rate by negligible amount.

Throughout this section, all channels are memoryless and rates are measured using base-$q$ units (or logarithms) for input and code alphabets satisfying $|\cX|=q$.

\begin{defn}[Symmetry group]
\label{def:standard_channel_symmetry_group}
Let $W(y\mid x)$ denote the transition probabilities of a channel with a finite input alphabet $\cX$ and a finite output alphabet $\cY$.
A permutation $\varphi \in \sym{\cX}$ is in the symmetry group of $W$ if there is a $\sigma \in \sym{\cY}$ such that $W(\sigma(y) \mid \varphi(x)) = W( y \mid x)$ for all $x\in \cX$ and $y\in \cY$.
A channel is called \emph{highly symmetric} if its symmetry group acts doubly-transitively on $\cX$. 
\end{defn}

\begin{defn}[Affine general linear symmetry]
\label{def:affine_symmetric_channel}
Let $\cX=\F_{q}$ where $q=r^b$ is a prime power.
For any invertible matrix $A \in \mathbb{F}_{r}^{b\times b}$ and vector $\bm{\beta} \in \mathbb{F}_r^b$, the map $\varphi(\bx)=A\bx+\bm{\beta}$ is an affine bijection on $(\F_r)^b \equiv \F_q$.
We say channel $W(y \mid x)$ is $\mathrm{AGL}(b)$\emph{ symmetric} if, for every such $\varphi(\bx)$ acting on $\F_q$, there exists a permutation $\sigma_\varphi$ of $\cY$ such that $W(\sigma_\varphi(y)\mid \varphi(x))=W(y\mid x)$ for all $x\in\F_q$ and $y\in\cY$.
\end{defn}

Note that channels that all channels with $\mathrm{AGL}(b)$ symmetry are highly symmetric because the affine group acts doubly transitively on $(\F_r)^b$.

\begin{defn} \label{def:full_rank}
  A discrete memoryless channel (DMC) is called \emph{injective} if it has no redundant inputs (i.e., the rows of its transition probability matrix $W(y\mid x)$ are linearly independent). 
  An indecomposable FSC is called \emph{injective} if, for large enough $n$, the transition probability matrix $W(\by \mid \bx)$ implied by $n$ channel uses, starting from the stationary state distribution $\pi$, defines an injective DMC.
\end{defn}

\begin{defn}
For $i\in [n]$, let $\varphi_i \colon \mathbb{F}_2^b \to \mathbb{F}_2^b$ denote independent random affine transformations defined by $\varphi_i (\bx) = A_i \bx + \bm{\beta}_i$ with uniform invertible $A_i \in \mathbb{F}_2^{b\times b}$ and uniform $\bm{\beta}_i \in \mathbb{F}_2^b$.
A \emph{random block scrambler} $\Phi \colon (\mathbb{F}_{2}^b)^n \to (\mathbb{F}_{2}^b)^n$ is the componentwise mapping defined by $\bz = \Phi(\bx ')$ where $\bz, \bx' \in (\mathbb{F}_{2}^b)^n$ and $\bz_i = \varphi_i (\bx_i)$ for $i\in [n]$.
\end{defn}

\begin{thm} \label{thm:capacity_symmetry_subcode}
  For $h\in \mathbb{N}$ and $b=2^h$,
  let $W (y \mid x)$ be a symmetric injective channel with $\cX = \mathbb{F}_2^b$ and $q=|\cX|$.
  Let $C$ be the capacity of $W$ (in qits per channel use).
  Then, for any $R < C$, there exists a binary code sequence $\cC_m = \RM(r_m , m)$ with rate sequence $R_m \to R$ (measured in bits per channel use) such that blocking the code bits into $b$-bit blocks and applying a random block scrambler gives a communication protocol with vanishing symbol error probability as $m\to \infty$.
\end{thm}
\ifarxiv
 \begin{proof}
   See Appendix~\ref{app:general_capacity_thm}.
 \end{proof}
\fi

\section{RM codes on FSCs via Block Interleaving}
\label{sec:rm_fsc_decimation}

\subsection{Block Interleaving for FSCs}

Now, we introduce two channel transformations for FSCs that will be useful.
The first transform changes an FSC into a new FSC that is equivalent to $b$ uses of the first channel.

\begin{defn}
    \emph{Blocking by $b$} transforms an FSC denoted by $(\cX,\cY,\cS,P)$ into an FSC denoted by $(\cX^b,\cY^b,\cS,P^b)$, where 
\begin{align*}
P^b_{Y, S'\mid X,S} & (\by,s_b \mid \bx,s_0) \\ & \coloneqq \sum_{s_1,\ldots,s_{b-1} \in \cS} \prod_{i=0}^{b-1} P_{Y, S'\mid X,S}(s_{i+1},y_i\mid x_i,s_{i}).
\end{align*}
    For a state distribution $\pi$ that is stationary under uniform inputs,
     let $\overline{W}^b (\by \mid \bx)$ denote the \emph{stationary average block channel} (for uniform inputs) defined by
    \[ \overline{W}^b (\by \mid \bx) \coloneqq \sum_{s,s' \in \cS} \pi(s) P^b_{Y, S'\mid X,S} (\by,s' \mid \bx,s).  \]
\end{defn}

\begin{rem}
    The blocking transformation preserves the symmetric capacity of an IFSC.
    This is because, for the original channel, the rate is measured in base-$|\cX|$ units and, for the blocked channel, the rate is measured in base-$|\cX|^b$ units.
    This cancels the factor of $b$ that comes from additional uses of the original channel.
\end{rem}

The second transform changes an FSC into a new FSC that is decimated by $d$, assuming uniform inputs.
Then, each input is transmitted through the original FSC followed by $d-1$ uniform random inputs.

\begin{defn}
    \emph{Decimating by $d$} transforms an FSC denoted by $(\cX,\cY,\cS,P)$ into an FSC denoted by $(\cX,\cY,\cS,P^{\downarrow d})$, where
\begin{align*}
P^{\downarrow d}_{Y', S'\mid X',S} & (y,s' \mid x,s) \\ & \coloneqq  \sum_{t\in \cS} P_{Y, S'\mid X,S} (y, t \mid x,s) \overline{\Phi}^{d-1}_{t,s'},
\end{align*}
where $\overline{\Phi}_{s,s'} = \sum_{x\in \cX} \frac{1}{|\cX|} \sum_{y\in \cY}  P_{Y, S'\mid X,S} (y, s' \mid x,s)$.
\end{defn}

Let $P^{b\downarrow d}_{Y,S'|X,S}$ denote the channel law of an FSC that is first blocked by $b$ and then decimated by $d$.
For the coding scheme described in Theorem~\ref{thm:capacity_symmetry_subcode},
the induced channel from $(\cX^b)^n$ to $(\cY^b)^n$, $n \in \mathbb{N}$, is 
\[
W^{b\downarrow d}_n (\by \mid \bx) = \!\! \sum_{s_0, \dots, s_n} \!\! \pi(s_0) \prod_{i=0}^{n-1} P^{b\downarrow d}_{Y,S'|X,S}(y_i, s_{i+1} \mid x_i, s_{i}). \]
Let $n$ uses of the memoryless channel $\overline{W}^b$ be denoted by $\overline{W}_n^b(\by \mid \bx) = \prod_{i=1}^n \overline{W}^b (\by_i \mid \bx_i)$.

If the FSC is indecomposable, then the following lemma shows that the distance between $W^{b\downarrow d}_n (\by \mid \bx)$ and $W^{b\downarrow d}_n (\by \mid \bx)$ decays exponentially in $d$.
A subtle point is that, while the argument seems to require a primitive FSC, starting an IFSC in its unique stationary distribution prevents it from visiting transitive states and effectively converts it into a primitive FSC.

\begin{lem} \label{lem:chan_tv_bound}
For any IFSC there is a $\rho \in (0,1)$ such that for all $b,d,n \in \mathbb{N}$, and $\bx \in (\cX^b)^n$, we have
\[ \| W^{b\downarrow d}_n (\cdot \mid \bx) - \overline{W}_n^b (\cdot \mid \bx)\|_{\mathrm{TV}} \leq (n-1)\rho^{b(d-1)}. \]
\end{lem}

\ifarxiv
 \begin{proof}
   See Appendix~\ref{sec_lem_chan_tv_bound}.
 \end{proof}
\fi

\subsection{RM Codes and Implicit Block Interleaving}
\label{sec:implicit_block_int}

For integers $h\ge 1$ and $g\ge 0$, we view the codeword coordinates of $\cC = \RM(r,m)$ as indexed by evaluation points $\bv\in\F_2^m$ via the bijection $\theta_m$. Decompose
\[
\bv = (\bu,\bz,\bw)\in \F_2^h \times \F_2^g \times \F_2^{m-h-g},
\]
where $\bu$ indexes \emph{within-block} positions, $\bz$ will be fixed, and $\bw$ indexes \emph{block locations}.
Since $\bz$ is fixed for each phase, this has the effect of restricting the evaluation set from $\F_2^m$ to the affine subspace
\[
\cV_{h,g} (\bz)
\ \triangleq\
\big\{(\bu,\bz,\bw):\ \bu\in\F_2^h,\ \bw\in\F_2^{m-h-g}\big\}.
\]
Let $J(\bz) = \theta_m \big(\cV_{h,g} (\bz)\big)$.
Then, evaluating all $f\in\mathcal{F}(r,m)$ on $\cV_{h,g} (\bz)$ yields a codeword in $\cC|_{J(\bz)}$ which is a punctured subcode naturally identified with $\RM(r_m,m-g)$ of length $2^{m-g}$. Moreover, the restricted coordinates decompose into
\[
n \ \triangleq\ 2^{m-h-g}
\]
disjoint \emph{blocks} of length $b=2^h$, one block for each $\bw\in\F_2^{m-h-g}$:
\[
\text{block } \bw:\quad \big(f(\bu,\bz,\bw)\big)_{\bu\in\F_2^h}\in\F_2^{2^h}.
\]
Under the indexing induced by $\theta_m$ and the above $(\bu,\bz,\bw)$ decomposition, concatenating these blocks in the natural binary order of $\bw$ yields exactly the punctured RM code $\cC|_{J(\bz)}$.

\paragraph*{Inner-code viewpoint and the block alphabet}
For fixed $\bw$, the map $\bu\mapsto f(\bu,0^g,\bw)$ is a Boolean polynomial of degree at most $r_m$ in $u_0,\ldots,u_{h-1}$. Hence, each block belongs to the length-$2^h$ code $\RM(r_m,h)$ with the convention $\RM(r_m,h)=\F_2^{2^h}$ is the full space if $r_m\ge h$.
Since this condition holds for all sufficiently large $m$ (e.g., $h=o(m)$ is chosen to grow arbitrarily slowly with $m$ and $r_m$ must grow linearly with $m$ to have positive rate in the limit), we see that the blocks themselves are unconstrained.

\vspace{2mm}
\subsection{Main Result}

\begin{thm}[RM w/scrambling achieves $C_{\mathrm{sym}}$ on IFSCs]
\label{thm:main}
Consider a binary-input IFSC that is full rank.
Let $C_{\mathrm{sym}}$ denote its symmetric capacity in bits per channel use.
Fix any target rate $R<C_{\mathrm{sym}}$.
Let the functions $h\colon \mathbb{N} \to \mathbb{N}$ and $g\colon \mathbb{N} \to \mathbb{N}$ satisfy, as $m\to \infty$, the conditions $h(m)\to \infty$, $h(m)/m \to 0$, $g(m)\to \infty$, and $g(m)/\sqrt{m} \to 0$.
Then, there exists a sequence of RM codes $\RM(r_m,m)$ with rates $R_m\to R$, such that the following communication strategy has vanishing bit error probability as $m\to\infty$:
\begin{enumerate}
\item Encode the message and partition the resulting RM codeword into blocks.
Then, apply independent random affine scrambling, which is known at the receiver, to each block and transmit the codeword.
 
\item After receiving, for each fixed $\bz$, decode the punctured subcode with codeword indices in $\cV_{h,g}(\bz)$ using MAP decoding for the induced (approximately independent) stationary block observations.

\item Combine decoded interleaves to get the whole codeword.
\end{enumerate}
\end{thm}

\ifarxiv
 \begin{proof}
   See Appendix~\ref{sec:proof_thm_main}.
 \end{proof}
\fi

\section{Conclusion and Open Problems}
\label{sec:conclusion}

We show that binary Reed--Muller codes can achieve rates up to the symmetric capacity on finite-state channels by combining three ideas: (i) a general capacity via symmetry theorem for memoryless channels based on code/channel symmetry and a puncturing property; (ii) block lifts that create well-separated protected blocks whose boundary states are nearly independent by exponential mixing; and (iii) per-block random affine scrambling that enforces strong symmetry of the induced block channel.

Some interesting questions and open problems include:
\begin{itemize}
\setlength{\itemsep}{0.5mm}
\item \emph{Removing scrambling.} Can one exploit intrinsic symmetries of the induced stationary block channel (or concentration) to eliminate per-block random affine scrambling?
-- We note that scrambling is currently used both to guarantee sufficient channel symmetry (via group symmetrization) and to break correlation between interleaves (i.e., channel inputs for other interleaves are treated as i.i.d.\ uniform).

\item \emph{Extending to block error.}
Can this result be improved to show vanishing block error rate for binary-input IFSCs? \\
-- It is probably possible to improve the symbol-error decay rate by extending approaches from~\cite{Abbe-focs23,Pfister-arxiv25a} to non-binary codes satisfying symmetry and nesting conditions. 
For interleaved binary RM codes, it seems likely that the weight enumerator argument will also extend.
One caveat is that Theorem~\ref{thm:capacity_symmetry_subcode_app} does not naturally extend to block-error rate without new insights or weight enumerator bounds for the non-binary code.

\item \emph{Input shaping and constraints.} Can one extend this approach to include shaping of the input distribution and/or FSCs with input constraints? \\
-- A standard approach to shaping is to add a distribution shaping inner code that converts a uniform input distribution into the desired non-uniform distribution.
For FSCs, this was proposed and analyzed in~\cite{Soriaga-isit04}.
Such methods should be compatible with our approach.
For input-constrained FSCs, our method does not work because it always produces a uniform input distribution.
But, one can overcome this by designing a distribution shaping inner code that converts a uniform input distribution into a distribution supported on the constrained set.
We are not aware of any obstruction to such schemes approaching capacity but a proper analysis wolud be required to make stronger statements.

\item \emph{Efficient decoding.} Our analysis is information-theoretic and uses MAP decoding on induced channels. Establishing analogous results with low-complexity decoders (e.g., recursive RM decoders) on FSCs is an important goal.

\item \emph{Non-binary alphabets.} This construction should extend naturally to FSCs with $q$-ary inputs.
\end{itemize}

\newpage

\bibliographystyle{ieeetr}
\bibliography{WCLabrv,WCLbib,WCLnewbib}

\newpage

\appendices

\section{Proof of Lemma~\ref{lem:ind2prim}}

\label{sec:proof_lem_ind2prim}

\begin{proof}[Proof of Lemma~\ref{lem:ind2prim}]
Let $\overline{\Phi}$ be the uniform state-transition matrix and let
\[
E \coloneqq \{(s,s')\in\cS^2:\ \overline{\Phi}_{s,s'}>0\}
\]
be the directed edge set of its support graph. Note that $\overline{\Phi}_{s,s'}=0$ implies
\[
0=\overline{\Phi}_{s,s'}=\frac{1}{|\cX|}\sum_{x\in\cX} P_{S'|X,S}(s'\mid x,s),
\]
and since all terms are nonnegative, this forces $P_{S'|X,S}(s'\mid x,s)=0$ for all $x\in\cX$.

A necessary and sufficient condition for an FSC to be indecomposable is the following \cite[Theorem~4.6.3]{Gallager-1968}:
There exists an $L > 0$ such that for any input sequence $\bx^L = (x_0,\ldots,x_{L-1})$, there is an $s_L \in \cS$ (which may depend on $\bx^L$) such that
\begin{equation}
\label{indecomp_equiv_cond}
\Pr(S_L = s_L \mid S_0 = s_0, X_0 = x_0,\ldots,X_{L-1} = x_{L-1}) > 0
\end{equation}
holds for all $s_0 \in \cS$.

\medskip\noindent
\emph{Step 1: $\overline{\Phi}$ has a unique recurrent class.}
Suppose for the sake of contradiction that $\overline{\Phi}$ has two disjoint closed communicating classes
$C_1$ and $C_2$.
Here, closed means $\overline{\Phi}_{s,s'}=0$ for all $s\in C_i$ and $s'\notin C_i$. Pick any $s \in C_1$ and $s' \in C_2$. Since the channel is indecomposable, there is an $L > 0$ such that for any $\bx^L = (x_0,\ldots,x_{L-1})$, there is an $s_L \in \cS$ such that \eqref{indecomp_equiv_cond} holds for both $s_0 = s$ and $s_0 = s'$.
With $C_1$ and $C_2$ closed, this implies that $s_L$ is in both $C_1$ and $C_2$.
But, this contradicts the initial assumption that $C_1$ and $C_2$ are disjoint.

Thus, $\overline{\Phi}$ has a unique closed
communicating class which we denote by $\cS_\star\subseteq \cS$.
For the Markov chain with transition matrix $\overline{\Phi}$,
all states in $\cS_\star$ are recurrent and all states in $\cS\setminus \cS_\star$ are transient.
Moreover, this Markov chain is irreducible and thus has a unique stationary distribution $\pi$.

\medskip\noindent
\emph{Step 2: The class $\cS_\star$ is aperiodic and $\overline{\Phi}$, when restricted to $\cS_\star$, is primitive.}
We remove We restrict to $\cS_\star$ by removing all (transient) states in $\cS\setminus \cS_\star$.
Suppose for the sake of contradiction that the chain restricted to $\cS_\star$ has period $d\ge 2$.
Then, there exists a cyclic partition $\cS_\star=D_0\cup\cdots\cup D_{d-1}$ such that
$\overline{\Phi}_{s,s'}>0$ implies $s\in D_i$ and $s'\in D_{i+1\ (\mathrm{mod}\ d)}$.
In particular, $\overline{\Phi}_{s,s'}=0$ whenever $s\in D_i$ and $s'\notin D_{i+1}$, and thus
$P_{S'|X,S}(s'\mid x,s)=0$ for all $x$ for such forbidden pairs $(s,s')$.
Consequently, for every input sequence $(x_0,\ldots,x_{L-1})$, the state must satisfy
\[
S_0\in D_i \ \Longrightarrow\ S_L\in D_{i+L \!\!\!\!\! \pmod{d}}.
\]
The argument now proceeds as in Step~1: Pick any $s \in D_0$, $s' \in D_1$. Then the indecomposability condition \eqref{indecomp_equiv_cond} implies that $s_L$ is simultaneously in $D_{L\ \!\!\!\!\! \pmod{d}}$ and $D_{L+1\ \!\!\!\!\! \pmod{d}}$, which is a contradiction.
Hence the period must be $1$, i.e., the class $\cS_\star$ is aperiodic.
Since the Markov chain is irreducible with recurrent class $\cS_\star$, the restriction of $\overline{\Phi}$ has period 1 and is therefore primitive.

\medskip\noindent
\emph{Step 3: Pruning to $\cS_\star$ does not change the symmetric capacity.}
Let $\bX^n$ be an i.i.d.\ uniform input process. For indecomposable FSCs, the limit in $C_{\mathrm{sym}} := \lim\limits_{n\to\infty} \frac{1}{n} I(\bX^n;\bY^n \mid S_0 = s_0)$ exists and does not depend on the initial state $s_0 \in \cS$ \cite{Blackwell-annmathstats58,Gallager-1968}. In particular, $s_0$ can be taken to be any state in $\cS_{\star}$. Since $\cS_{\star}$ is a closed communicating class of $\overline{\Phi}$, if the initial state of the FSC is in $\cS_{\star}$, then all subsequent states are in $\cS_{\star}$ as well. Hence, the symmetric capacity of the original FSC is the same as that of the pruned FSC.
\end{proof}

\section{Proof of Lemma~\ref{lem:chan_tv_bound}}

\label{sec_lem_chan_tv_bound}

\begin{proof}
Fix $b,d,n\in\mathbb{N}$ and an arbitrary input sequence
$\bx=(x_0,\ldots,x_{n-1})\in(\cX^b)^n$.
We compare two conditional output laws on $(\cY^b)^n$: the true blocked decimated FSC denoted by (A) and the memoryless block channel denoted by (B).
A key observation is that the two differ only through the block boundary state dependence.

In the true channel, each output block $Y_i$ is generated by $b$ uses of the original channel.
This implies that the uniform state-transition matrix of the blocked channel equals $\overline\Phi^{b}$ where $\overline\Phi$ is the uniform state-transition matrix of the original channel.
The next $d-1$ input blocks correspond to $b(d-1)$ uses of the original channel which do not affect $Y_i$ but only affect the next boundary state $S_{i+1}$ through the matrix $\overline\Phi^{b(d-1)}$.

Consequently, for each $i$ and each starting state $s$,
the conditional law of $Y_i$ given $(X_i=x_i,S_i=s)$ is the same in (A) and (B);
the difference is that in (A) the sequence $S_i$ is a Markov chain controlled by the input,
whereas in (B) the state sequence $\widetilde{S_i}$ is i.i.d.\ from $\pi$.
We use different notation for the state sequences in these two cases because it is best to think of them as two different random variables defined on the same probability space where they are maximally coupled.

Since the original FSC is indecomposable, Lemma~\ref{lem:ind2prim} implies (after restricting to the
unique recurrent class) that $\overline\Phi$ is primitive with stationary distribution $\pi$.
In this setup, the restriction to the unique recurrent class for $W^{b \downarrow d}_n$ and $\overline{W}_n^b$ occurs implicitly by drawing the initial channel state from the unique stationary distribution.
This initialization restricts the state to the unique recurrent class because the unique stationary distribution is supported on that set.
Moreover, the distribution is stationary and the probability of visiting a transient state remains equal to zero.

Since the chain can now be treated as primitive, Theorem~\ref{thm:exp_mixing} shows that there exists $\rho\in(0,1)$ (depending only on the FSC through $\overline\Phi$) such that
for every probability measure $\mu$ on $\cS$ and every $t\in\mathbb{N}$,
\begin{equation}\label{eq:tv-mixing}
\big\|\mu\,\overline\Phi^{t}-\pi\big\|_{\mathrm{TV}}\ \le\ \rho^{t}.
\end{equation}
Now, we build a coupling between the outputs under (A) and (B) given the same input $\bx$.
Draw $S_0$ from $\pi$ and set $\widetilde{S}_0 = S_0$.
Given a common boundary state $S_i=\widetilde{S}_i=s$, generate the $i$-th output block identically
in both experiments by using the same randomness for the $b$ uses of the original channel.
This is possible
because the conditional law of $Y_i$ given $(X_i=x_i,S_i=s)$ is the same in (A) and (B).
Let $T_i$ denote the state at the end of these $b$ channel uses in experiment (A)
(i.e., this is the terminal state of the blocked kernel).
Then, the next true boundary state satisfies
\[
\Pr(S_{i+1}=s \mid T_i = t)\;=\;\overline\Phi^{b(d-1)}_{t,s}
\]
because the blocking and decimation results in $b(d-1)$ uniform-input transitions of the original channel.

In case (B), the next block-start state is drawn
$\widetilde{S}_{i+1}\sim \pi$, independent of the past.
Now, we can couple $S_{i+1}$ and $\widetilde{S}_{i+1}$ by maximal coupling conditional on $T_i$.
Under such a coupling,
\[
\Pr(S_{i+1}\neq \widetilde{S}_{i+1}\mid T_i)
\;=\;\big\|e_{T_i}\,\overline\Phi^{b(d-1)}-\pi\big\|_{\mathrm{TV}}
\;\le\;\rho^{b(d-1)},
\]
where $e_s \in \mathbb{R}^{|\cS|}$ is the canonical basis vector for $s$ and the inequality follows from \eqref{eq:tv-mixing} with $t=b(d-1)$ and $\mu= e_{T_i}$.
This bound holds uniformly over $T_i$ and thus also unconditionally.

If $S_{i+1}=\widetilde{S}_{i+1}$, we proceed to the next block and again generate the next output block
identically; otherwise, we allow the two experiments to evolve independently thereafter.
Therefore, the probability that the coupled output sequences differ anywhere is at most
the probability that a mismatch occurs at one of the $n-1$ transitions between blocks
\begin{align*}
\Pr((Y_0, \scalebox{0.9}{\ldots},Y_{n-1})\neq (\widetilde{Y}_0, \scalebox{0.9}{\ldots}, \widetilde{Y}_{n-1}))
 & \le\ \sum_{i=0}^{n-2}\Pr(S_{i+1}\neq \widetilde{S}_{i+1}) \\
 & \le\ (n-1)\rho^{b(d-1)}.
\end{align*}
Finally, using the coupling characterization,
the total variation distance between the two conditional output laws is upper bounded by the
minimal disagreement probability over all couplings, and hence by the disagreement probability
of the coupling constructed above:
\[
\big\| W^{b\downarrow d}_n(\cdot\mid \bx)\;-\;\overline W_n^b(\cdot\mid \bx)\big\|_{\mathrm{TV}}
\ \le\ (n-1)\rho^{b(d-1)}.
\]
This completes the proof.
\end{proof}

\section{Proof of Theorem~\ref{thm:main}}

\label{sec:proof_thm_main}

\begin{proof}
Fix a binary-input IFSC $(\cX,\cY,\cS',P')$ with $\cX=\{0,1\}$ and assume it is full rank.
By Lemma~\ref{lem:ind2prim}, the uniform state-transition probability matrix $\overline{\Phi}'$ has a unique stationary distribution which we denote by $\pi'$.
After sending a sufficiently long warmup sequence, we can assume the channel is in a state with positive probability.
Thus, we can restrict the channel to the states with positive probability to get a primitive FSC $(\cX,\cY,\cS,P)$ on a smaller state space with the same symmetric capacity.
For this reason, we restrict our attention to primitive FSCs.

\paragraph{Step 1: Choose block and spacing parameters}
Let $m\to\infty$ and choose integers
\[
h=h(m)\in\mathbb{N},\qquad g=g(m)\in\mathbb{N}.
\]
Set the blocking factor $b$, decimation $d$, and inter-block spacing $T$ to be
\[
b \ \triangleq\ 2^{h},\qquad d \ \triangleq\ 2^{g},\qquad T \ \triangleq\ bd \ =\ 2^{h+g}.
\]
For each fixed $\bz\in\F_2^g$, the coordinate set $\cV_{h,g}(\bz)$ of that interleave consists of
\[
n \ \triangleq\ 2^{m-h-g}
\]
blocks, each of length $b=2^h$, and consecutive blocks of this interleave (denoted by $\bz$) begin $T$ channel uses apart.

\paragraph{Step 2: RM restriction produces protected subcodes}
We index coordinates of $\RM(r_m,m)$ by $\bv=(\bu,\bz,\bw)\in \F_2^h\times\F_2^g\times\F_2^{m-h-g}$ via $\theta_m$.
For each fixed $\bz\in\F_2^g$, restricting to $\cV_{h,g}(\bz)$ yields the punctured subcode of length $2^{m-g}$ that is naturally identified with $\RM(r_m,m-g)$.
Grouping its coordinates into blocks indexed by $\bw\in\F_2^{m-h-g}$ gives a length-$n$ group code over the block alphabet $\cX_{\mathrm{blk}}\ \triangleq\ \F_2^{b}$.
Here, we assume the inner-code is unconstrained as described in Section~\ref{sec:implicit_block_int}.
Thus, for large enough $m$, we assume each block can realize any element of $\F_2^{b}$.

Moreover, by Lemma~\ref{lem:rm_affine_sym}, the automorphism group of $\RM(r_m,m-g)$ contains all invertible affine maps on its evaluation space $\F_2^{m-g}$.
In particular, it contains affine maps that act only on the $\bw$-coordinates (and leave $\bu$ fixed), which act doubly transitively on the $n$ block indices.
Thus, for each $\bz$, the block-level code is a doubly-transitive group code of length $n$ over $\F_2^{b}$.

\paragraph{Step 3: Define the induced block channel}
Fix $\bz$.
Under the encoding rule, the transmitted sequence consists of the $n$ protected blocks in the interleave which is indexed by $\bz$.
Each of these blocks is separated by $b(d-1)$ interstitial symbols (belonging to the other $\bz'\neq \bz$ phases).
Because every block $(\bz',\bw')$ is independently scrambled by a uniformly random affine function, the interstitial inputs can be treated as i.i.d.\ $\mathrm{Bernoulli}(1/2)$ from the perspective of decoding the $\bz$-interleave.
Equivalently, the boundary-state evolution between two consecutive protected blocks for this fixed $\bz$ is driven by the \emph{uniform state-transition matrix} $\overline{\Phi}$.

Therefore, the effective channel seen by the $\bz$-decoder is the blocked-and-decimated FSC with block size $b$ and decimation factor $d$.
The true channel is denoted by
\[
W^{b\downarrow d}_n:\ (\F_2^{b})^n \to (\cY^{b})^n,
\]
with stationary initial distribution $S_0\sim\pi$ and kernel $P^{b\downarrow d}_{Y,S'|X,S}$ as defined in Section~\ref{sec:rm_fsc_decimation}.
Let $\overline W^b$ denote the stationary average \emph{one-block} channel (a DMC on $\F_2^b\to\cY^b$) and $\overline W_n^b=(\overline W^b)^{\otimes n}$ its $n$-fold product.

By Lemma~\ref{lem:chan_tv_bound}, since the FSC is primitive (hence indecomposable) there exists $\rho\in(0,1)$ such that for all $\bx\in(\F_2^b)^n$,
\begin{equation}\label{eq:tv-bound-proof-main}
\big\|W^{b\downarrow d}_n(\cdot\mid\bx) - \overline W_n^b(\cdot\mid\bx)\big\|_{\mathrm{TV}}
\ \le\ (n-1)\rho^{b(d-1)}.
\end{equation}

\paragraph{Step 4: Capacity of the stationary block channel approaches $C_{\mathrm{sym}}$}
Let $\bX^{b}$ be i.i.d.\ $\mathrm{Bernoulli}(1/2)$ and let $\bY^{b}$ be the FSC output over $b$ consecutive uses when $S_0\sim\pi$.
Then, the mutual information of the stationary block channel $\overline W^b$ under uniform inputs satisfies
\[
I_{\overline W^b}(\bX^{b};\bY^{b}) \ =\ I(\bX^{b};\bY^{b}\mid S_0\sim\pi).
\]
By the definition of $C_{\mathrm{sym}}$ and standard IFSC information-rate limits for stationary initialization,
\begin{equation} \label{eq:block_rate_gap}
\lim_{b\to \infty} \frac{1}{b} I(\bX^{b};\bY^{b}\mid S_0\sim\pi)\ = C_{\mathrm{sym}}.
\end{equation}

Hence, if $h(m) \to \infty$, then $b=2^{h(m)} \to \infty$ and, for any target rate $R< C_{\mathrm{sym}}$, there is an $m_0>0$ such that, for all $m>m_0$, the mutual information is greater than $R$.

\paragraph{Step 5: Per-block scrambling yields a highly symmetric memoryless block channel}
Augment the stationary one-block channel to include the scrambling operation in the output, i.e.,
\[
\widetilde W^b:\ \F_2^{b} \ \to\ \cY^{b}\times \mathrm{AGL}(b,2),
\]
is defined by drawing $\varphi\sim\mathrm{Unif}(\mathrm{AGL}(b,2))$ and preprocess the input by applying $\varphi$.
The preprocessed input is used to drive $b$ consecutive uses of the FSC started at $S_0\sim\pi$ and then the channel outputs $(Y^{b},\varphi)$.
Since $\varphi$ is revealed and does not change the input distribution, this augmentation does not decrease mutual information, and the induced memoryless $n$-block channel is $(\widetilde W^b)^{\otimes n}$.

By construction, $\widetilde W^b$ is highly symmetric with respect to the (doubly-transitive) action of $\mathrm{AGL}(b,2)$ on the input symbols.
In addition, the induced block channel is full rank for all sufficiently large $b$ if the FSC is full rank (follows Definition~\ref{def:full_rank} with $n=b$.).

\paragraph{Step 6: Vanishing error on the memoryless block channel for each interleave}
Fix $\bz$.
Consider the protected subcode on $\cV_{h,g}(\bz)$, viewed as a length-$n$ doubly-transitive group code over $\F_2^{b}$, transmitted over the memoryless channel $(\widetilde W^b)^{\otimes n}$.
Now, we can apply Theorem~\ref{thm:capacity_symmetry_subcode} to show that this code achieves capacity because the channel is full rank and highly symmetric and the code family is doubly transitive with the required puncturing property.
Thus, any rate strictly below the symmetric capacity of $\widetilde W^b$ is achievable with vanishing symbol error probability under MAP decoding of a block-scrambled RM code on the memoryless block channel.

Now choose a sequence $\RM(r_m,m)$ with rates $R_m\to R$ in bits per use.
Since the restriction to $\cV_{h,g}(\bz)$ (caused by decimation) increases the RM rate by at most $o(1)$ when $g=o(\sqrt m)$ (cf.\ Lemma~\ref{lem:rm_rate_diff} and the standard restriction/decimation argument), the effective rate of each $\bz$ interleave is at most $R+\eta/2$ for all large enough $m$.
Together with \eqref{eq:block_rate_gap}, this shows the rate of the restricted code is strictly below the symmetric capacity of $\widetilde W^b$.
Thus, the symbol error probability of the MAP decoder for the $\bz$-subcode on $(\widetilde W^b)^{\otimes n}$ satisfies
\[
\widetilde{P}^{(\bz)}_{\mathrm{err}}(m) \, \longrightarrow \, 0.
\]

\paragraph{Step 7: Transfer from memoryless to the true induced channel using total variation}
Run the \emph{same} $\bz$-decoder on the true induced channel $W_n^{b\downarrow d}$ (with scramblings revealed and inverted as in the theorem statement).
Then, for any fixed decoder, the difference in its error probability under two channels is bounded by the total variation distance between the induced output laws (a standard consequence of the coupling characterization of TV).
Thus, for any transmitted sequence, the true symbol error probability satisfies
\[
P^{(\bz)}_{\mathrm{err}}(m)
\ \le\ 
\widetilde{P}^{(\bz)}_{\mathrm{err}}(m)\ +\ 
\sup_{\bx}\big\|W^{b\downarrow d}_n(\cdot\mid\bx)-\overline W_n^b(\cdot\mid\bx)\big\|_{\mathrm{TV}}.
\]
Invoking \eqref{eq:tv-bound-proof-main} shows that
\[
P^{(\bz)}_{\mathrm{err}}(m)
\ \le\ 
\widetilde{P}^{(\bz)}_{\mathrm{err}}(m)\ +\ (n-1)\rho^{b(d-1)}.
\]
Because $n=2^{m-h-g}$ grows at most exponentially in $m$ while $d=2^g\to\infty$ and $\rho^{d-1}$ decays double-exponentially in $g$, we may choose $g=g(m)$ (still with $g=o(\sqrt m)$) such that
\[
(n-1)\rho^{b(d-1)} \,\longrightarrow\, 0.
\]
Hence, for each fixed $\bz$, we have
\[
P^{(\bz)}_{\mathrm{err}}(m) \,\longrightarrow\, 0.
\]

\paragraph{Step 8: Approaching $C_{\mathrm{sym}}$}
Since $\frac{1}{b}I(\bX^b;\bY^b\mid S_0\sim\pi)\to C_{\mathrm{sym}}$ as $b=2^{h}\to\infty$, we may let $h(m)\to\infty$ sufficiently slowly so that the achievable block-level symmetric capacity exceeds any target $R<C_{\mathrm{sym}}$ for all large $m$, while simultaneously choosing $g(m)\to\infty$ to ensure approximate independence across protected blocks via Lemma~\ref{lem:chan_tv_bound}.
Thus, the achieved rate can be taken arbitrarily close to $C_{\mathrm{sym}}(W)$ and this completes the proof.
\end{proof}

\section{Capacity via Symmetry for memoryless channels}
\label{app:general_capacity_thm}

\subsection{Code symmetry and matching}

Let $\cX$ denote the channel input alphabet.
To define a group code for this channel, one can identify the elements of $\cX$ with an arbitrary group containing $|\cX|$ elements.
This choice is often determined by the desired structure of the code.
Another consideration is matching the group structure of the code to that of the channel.
Fortunately, the latter concern is not binding because one can always use the process of group symmetrization (e.g., see~\cite{Reeves-isit23}) to expand the symmetry group of the channel as needed.
Thus, the group structure chosen for $\cX$ is typically dictated by the code structure.

\begin{defn}[$q$-ary group code]
\label{def:transitive_group_code}
A code $\cC\subseteq \cX^N$ is a \emph{$q$-ary group code} if it is a subgroup of $(\cX^N,+)$.
\end{defn}

\begin{defn}[permutation-automorphism group]
\label{def:paut}
The \emph{permutation-automorphism group} of a code $\cC$ is the set of coordinate permutations that preserve the code
\[ \Aut(\cC) \coloneqq \{ \sigma \in \sym{N} \mid  \forall \bc \in \cC, (c_{\sigma(0)},\ldots,c_{\sigma(N-1)}) \in \cC \}. \]
We say that the code is (doubly) transitive if this permutation group is (doubly) transitive.
\end{defn}

\begin{defn}[symbol-automorphism group]
\label{def:saut}
The \emph{symbol-automorphism group} of a code $\cC$ is the set of homogeneous symbol relabelings that preserve the code
\[
\SAut(\cC) \coloneqq \{ \sigma \in \sym{\cX} \mid  \forall \bc \in \cC, (\sigma (c_0),\ldots \sigma (c_{n-1})) \in \cC \}. 
\]
We say that the code is \emph{symbol symmetric} if this permutation group is transitive and that it is \emph{highly symbol-symmetric} if this permutation group is doubly transitive.
\end{defn}

\begin{lem}[Uniform marginals]
\label{lem:uniform_marginals}
Let $\cC \subseteq \cX^N$ be a group code, and let $\bX$ be uniformly distributed on $\cC$. Then for each $i \in [N]$, the coordinate $X_i$ is uniformly distributed on the scalar projection $\cC|_{\{i\}} \subseteq\cX$. Moreover, if $\cC$ is symbol-symmetric, then $X_i \sim \mathrm{Unif}(\cX)$ for all $i \in [N]$.
\end{lem}

\begin{proof}
Since $\cC$ is a subgroup, $\cC|_{\{i\}}$ is a subgroup of $(\cX,+)$. A uniform $\bX$ on $\cC$ pushes forward through the projection so that $\cC|_{\{i\}}$ is uniform over $\cX$. If $\cC$ is symbol-symmetric then $\mathrm{proj}_i(X)$ equals $\cX$ for all $i\in [N]$.
\end{proof}

\begin{defn}
\label{def:matched}
A subgroup $\cG$ of the symbol-automorphism group of the code is \emph{matched to channel $W$} if the symmetry group of $W$ contains $\cG$ as a subgroup.
\end{defn}

\subsection{Capacity via symmetry}

The following theorem can be distilled from~\cite{Reeves-isit23}.
\begin{thm}[Capacity via symmetry]
\label{thm:capacity_symmetry_subcode_app}
Fix $q\ge 2$ and a memoryless channel $W(y \mid x)$ with symmetric capacity $C>0$ (base-$q$ units).
Let $\{\cC_n\}_{n\ge 1}$ be a sequence of group codes $\cC_n\subseteq\cX^{N_n}$ with $N_n\to\infty$ and rates $R_n\triangleq R(\cC_n)\to R$, where $0\le R<C$.

Assume that for all sufficiently large $n$:
\begin{enumerate}
\item[\textbf{(A1)}] (\emph{Permutation symmetry}) $\cC_n$ has a doubly-transitive permutation-automorphism group.

\item[\textbf{(A2)}] (\emph{Symbol symmetry}) $\cC_n$ has a symbol-automorphism group with a doubly-transitive subgroup that is matched to the channel $W$.

\item[\textbf{(A3)}] (\emph{Non-degeneracy}) The memoryless channel is injective.

\item[\textbf{(A4)}] (\emph{Two-look property})
There exists a subset $J_n\subseteq[N_n]$ with $0\in J_n$ and $|J_n|/N_n\to 0$ such that the punctured code $\cC_n' \triangleq \cC_n|_{J_n}$ has a transitive permutation automorphism group and
\[
\Delta(\cC_n,J_n):=\big|R(\cC_n')-R(\cC_n)\big|\ \xrightarrow[n\to\infty]{}\ 0.
\]
\end{enumerate}
Then, letting $X$ be uniformly distributed on $\cC_n$ and $Y$ be the output of $W^{\otimes N_n}$, the maximal symbol-error rate under symbol-MAP decoding satisfies
\[
\max_{i\in[N_n]} \mathrm{SER}(X_i\mid Y)\ \xrightarrow[n\to\infty]{}\ 0.
\]
In particular, $\{\cC_n\}$ achieves the symmetric capacity $C$ with respect to symbol-error probability.
\end{thm}

Theorem~\ref{thm:capacity_symmetry_subcode_app} does not apply directly to the blocked
RM code (defined below) obtained by grouping $b$ consecutive bits of a binary RM codeword into a single symbol of $\F_2^b$.
This is because the blocked code does not have a large enough symbol-automorphism group.
Instead, we will apply Theorem~\ref{thm:capacity_symmetry_subcode_app} to interleaved RM codes (defined below) and show that the symbol error rate vanishes for this family as the block length increases.
Since a blocked RM code is a subcode of an interleaved RM code, this implies that the symbol error rate of the blocked RM code also vanishes.
In addition, the rate difference between the two codes vanishes as the blocklength increases.
Thus, the blocked channel also achieves capacity.
The ideas outlined above are presented formally below.

\begin{defn}
\label{def:interleaved_rm_supercode}
Fix $h\in \mathbb{N}$, $b=2^h$, $m\ge h$, and $N = 2^{m-h}$.
The \emph{interleaved RM} (IRM) code $\cI_m \subseteq (\F_2^b)^N$ is the group code over $\cX=\F_2^b$ obtained by: (i) taking any $b$ codewords from $\RM(r_m,m-h)$ and stacking them into a $b\times N$ matrix and (ii) treating each matrix column as a code symbol in $\cX$.
\end{defn}

\begin{lem}[Blocked RM subcode of IRM code]
\label{lem:blocked_rm_subcode}
Fix $h\in \mathbb{N}$, $b=2^h$ and let the \emph{blocked RM} (BRM) code $\cB_m\subseteq(\F_2^b)^N$ be obtained by grouping consecutive blocks of $b$ bits from $\RM(r_m,m)$ into $\F_2^b$ symbols, where
$N=2^{m-h}$.
Then,
\[
\cB_m \subseteq \cI_m.
\]
\end{lem}

\begin{proof}
Write each evaluation point $\bv\in\F_2^m$ as
\[
\bv=(\bu,\bw)\in \F_2^h\times \F_2^{m-h},
\]
where $\bu$ corresponds to the $h$ least-significant coordinates under the map
$\theta_m$ from Definition~\ref{def:theta}.
Thus, consecutive groups of $b=2^h$ coordinates correspond exactly to fixing
$\bw$ and letting $\bu$ vary over $\F_2^h$.

Let $f\in\mathcal{F}(r_m,m)$ and let $\bc\in\RM(r_m,m)$ be its evaluation vector.
After grouping into blocks of size $b$, the resulting codeword in $(\F_2^b)^N$
can be viewed as a $b\times N$ matrix $M_f$ with rows indexed by $\bu\in\F_2^h$
and columns indexed by $\bw\in\F_2^{m-h}$, where
\[
M_f(\bu,\bw)=f(\bu,\bw).
\]
For each fixed $\bu$, the row function
\[
\bw \mapsto f(\bu,\bw)
\]
is a polynomial in $m-h$ variables of degree at most $r_m$ and thus the row belongs to $\RM(r_m,m-h)$.
Therefore, every blocked RM codeword belongs to $\cI_m$, proving
$\cB_m\subseteq \cI_m$.
\end{proof}

\begin{lem}
\label{lem:blocked_rm_rate_gap}
The rate gap between $\cI_m$
and $\cB_m$ is non-negative and, if $h=h(m)=o(\sqrt m)$, then
\[
R(\cI_m)-R(\cB_m)\to 0.
\]
\end{lem}

\begin{proof}
Grouping $b$ consecutive binary coordinates into one symbol of $\F_2^b$ does not
change the numerical rate when the latter is measured in base-$2^b$ units.
Hence, we have
\[
R(\cB_m)=\frac{\dim(\RM(r_m,m))}{2^m}=R(\RM(r_m,m)).
\]
Likewise, $\cI_m$ consists of $b$ independent rows from $\RM(r_m,m-h)$.
Thus, we have
\[
|\cI_m|=|\RM(r_m,m-h)|^b
      =2^{\,b\dim(\RM(r_m,m-h))}.
\]
Computing its rate over the alphabet $\F_2^b$ gives
\[
R(\cI_m)
=
\frac{b\,\dim(\RM(r_m,m-h))}{b\,2^{m-h}}
=
R(\RM(r_m,m-h)).
\]
The final claim follows from Lemma~\ref{lem:rm_rate_diff} with $k=h$.
\end{proof}

\begin{lem}
\label{lem:subcode_ser_monotone}
Let $\cD, \cE$ be group codes over $\F_2^b$ satisfying $\cD\subseteq\cE\subseteq(\F_2^b)^N$ and assume the channel is matched to the symbol-automorphism group of the code.
Then, the symbol-MAP error probability for $\cD$ is at most that of $\cE$.
In particular, for the codes above,
\[
\mathrm{SER}_{\cB_m} \le \mathrm{SER}_{\cI_m}.
\]
\end{lem}

\begin{proof}[Proof]
Consider a codeword drawn uniformly from $\cE$.
Let $U$ denote the coset of $\cD$ inside $\cE$ containing the transmitted codeword.
Conditioned on $U=u$, the transmitted codeword is uniform on a single coset
$u+\cD$.
Because the channel is matched to the symbol-automorphism group of $\cE$, all such cosets of $\cD$ in $\cE$ are equivalent under a known symbolwise translation, so the conditional symbol-MAP error given
$U=u$ is the same as for $\cD$ itself.
Revealing side information cannot increase MAP error, hence
\[
\mathrm{SER}_{\cE}
\ge
\mathrm{SER}_{\cE\mid U}
=
\mathrm{SER}_{\cD}.
\]
Applying this with $\cD=\cB_m$ and $\cE=\cI_m$ gives the claim.
\end{proof}

Now, we give a proof of Theorem~\ref{thm:capacity_symmetry_subcode}.

\begin{proof}[Proof of Theorem~\ref{thm:capacity_symmetry_subcode}]
Fix $h\in \mathbb{N}$, $b=2^h$, and let
\[
N_m \triangleq 2^{m-h}.
\]
Take the binary Reed--Muller code
\[
\RM(r_m,m)\subseteq \F_2^{2^m},
\]
group its coordinates into consecutive blocks of $b$ bits to get the BRM code
\[
\cB_m \subseteq (\F_2^b)^{N_m}.
\]
This code is contained in the IRM code $\cI_m\subseteq(\F_2^b)^{N_m}$ described in Definition~\ref{def:interleaved_rm_supercode}.

Our goal is to prove that $\cB_m$ achieves capacity.
Thus, we will apply Theorem~\ref{thm:capacity_symmetry_subcode_app} to $\cI_m$, and then
pass to the subcode $\cB_m$ using Lemma~\ref{lem:subcode_ser_monotone}.

\medskip
\noindent\textbf{1) \(\cI_m\) is a group code with rate approaching \(R\).}

By construction, $\cI_m$ is a group code over the alphabet $\F_2^b$.
By Lemma~\ref{lem:blocked_rm_rate_gap},
\[
R(\cI_m)=R(\RM(r_m,m-h)).
\]
Choose the binary RM sequence so that
\[
R(\cB_m)=R(\RM(r_m,m))\to R.
\]
If $h=h(m)=o(\sqrt m)$, then Lemma~\ref{lem:blocked_rm_rate_gap} gives
\[
R(\cI_m)-R(\cB_m)\to 0,
\]
and hence also
\[
R(\cI_m)\to R.
\]

\medskip
\noindent\textbf{2) Conditions \textbf{(A1)} and \textbf{(A2)} hold for \(\cI_m\).}

For (A1), let $\pi$ be any permutation automorphism of $\mathrm{RM}(r,m-h)$ and consider the matrix representation of a codeword of $\cI_m$ as described in Definition~\ref{def:interleaved_rm_supercode}.
Applying $\pi$ simultaneously to every row of the matrix sends
\[
(\bc^{(0)},\ldots,\bc^{(b-1)})
\mapsto
(\pi\bc^{(0)},\ldots,\pi\bc^{(b-1)}),
\]
which is again a valid codeword of $\cI_m$ because each $\pi\bc^{(j)}$ remains in
$\mathrm{RM} (r,m-h)$.
Thus, the permutation automorphism group of $\cI_m$ contains the
permutation automorphism group of $\mathrm{RM} (r,m-h)$ and is doubly transitive.

For (A2), let $A$ be an invertible $b\times b$ matrix over $\F_2$ and $\bm{\beta}\in\F_2^b$.
Applying the affine map $\sigma_{A,\bm{\beta}}:\F_2^b\to\F_2^b$, defined by
\[
\sigma_{A,\bm{\beta}}(u)=Au+\bm{\beta},
\]
to a codeword $\bx\in\cI_m$ symbolwise is equivalent to applying the linear map $A$ to each column in the matrix representation of $\bx$, followed by adding the constant vector $\bm{\beta}$.
Thus, if the rows of $\bx$ are $\bc^{(0)},\ldots,\bc^{(h-1)}$, then the transformed rows are
\[
\widetilde{\bc}^{(i)}
=
\sum_{j=0}^{b-1} A_{ij}\bc^{(j)} + \beta_i \bm{1},
\qquad i\in[h].
\]
Since $\mathrm{RM} (r,m-h)$ is linear and contains the all-ones codeword $\bm{1}$, each $\widetilde{\bc}^{(i)}$ again belongs to
$\RM (r,m-h)$.
Hence, the symbolwise affine map preserves $\cI_m$ and the symbol-automorphism group of the code contains $\mathrm{AGL}(b)$.
Using the idea of group symmetrization from~\cite{Reeves-arxiv23}, we see that the random scrambling operation transforms the channel $W(y|x)$ into a channel whose capacity $C$ equals the symmetric capacity $C_{\mathrm{sym}}$ of the original channel and whose symmetry group contains $\mathrm{AGL}(b)$.
Thus, the code contains a doubly-transitive symbol-automorphism group that is matched to the symmetry group of the channel.

\medskip
\noindent\textbf{3) Condition \textbf{(A3)}.}

Assumption \textbf{(A3)} is exactly the standing assumption that the channel $W(y\mid x)$ is injective.

\medskip
\noindent\textbf{4) Condition \textbf{(A4)} for \(\cI_m\).}

Take any integer sequence $g_m\to\infty$ such that
\[
g_m=o(\sqrt m).
\]
Using the standard RM affine-subspace restriction argument, puncturing
$\RM(r_m,m-h)$ to a codimension-$g_m$ affine subspace produces
\[
\RM(r_m,m-h-g_m).
\]
Applying this rowwise to $\cI_m$ gives a punctured code $\cI_m'$ on an index set
$J_m\subseteq[N_m]$ with $0\in J_m$ and
\[
\frac{|J_m|}{N_m}\to 0.
\]
Moreover, the same arguments used in~\cite{Reeves-arxiv23}
show that $\cI_m'$ has a transitive permutation-automorphism group.

Finally, since $R(\cI_m)=R(\RM(r_m,m-h))$ and $R(\cI_m')=R(\RM(r_m,m-h-g_m))$, we see that Lemma~\ref{lem:rm_rate_diff} gives
\[
|R(\cI_m')-R(\cI_m)|\to 0.
\]
Thus, the concrete puncturing condition used in the proof of the non-binary RM capacity theorem in~\cite{Reeves-arxiv23} implies \textbf{(A4)} here.

\medskip
\noindent\textbf{5) Apply Theorem~\ref{thm:capacity_symmetry_subcode_app} to $\cI_m$.}

We have verified \textbf{(A1)}--\textbf{(A4)} for the sequence $\{\cI_m\}$.
Since $R(\cI_m)\to R<C$, Theorem~\ref{thm:capacity_symmetry_subcode_app} implies
that the maximal symbol-error rate under symbol-MAP decoding on the channel
$W^{\otimes N_m}$ satisfies
\[
\max_{i\in[N_m]} \mathrm{SER}_{\cI_m}(X_i\mid Y)\to 0.
\]

\medskip
\noindent\textbf{6) Pass from $\cI_m$ to $\cB_m$.}

By Lemma~\ref{lem:blocked_rm_subcode}, we have
\[
\cB_m\subseteq \cI_m.
\]
Hence, by Lemma~\ref{lem:subcode_ser_monotone}, it follows that
\[
\max_{i\in[N_m]} \mathrm{SER}_{\cB_m}(X_i\mid Y)
\le
\max_{i\in[N_m]} \mathrm{SER}_{\cI_m}(X_i\mid Y).
\]
The right-hand side tends to zero, so the same is true for the blocked RM code.
This proves Theorem~\ref{thm:capacity_symmetry_subcode}.
\end{proof}

\section{Exponential Mixing of Markov Chain}

\begin{defn}
A finite-state Markov chain with transition matrix $P\in \mathbb{R}_{+}^{d\times d}$ is called \emph{irreducible} if, for any $i,j\in[d]$, there exists $k \in \mathbb{N}$ such that $[P^k]_{i,j} > 0$.
It is \emph{primitive} (i.e., irreducible and aperiodic) if there exists $k \in \mathbb{N}$ such that $[P^k]_{i,j} > 0$ for all $i,j\in[d]$.
\end{defn}

It is well-known that a primitive Markov chain has a unique stationary distribution $\pi \in \mathbb{R}_{++}^{d}$ such that $\pi P = \pi$.
Consider the vector space $V=\mathbb{R}^{d}$ of row vectors equipped with the norm $\left\Vert \bm v\right\Vert _{1}=\sum_{i=1}^{d}\left|v_{i}\right|$ and let
\[
U=\left\{ \bm x\in V\,\middle|\,x_{i}\geq0,\ \sum_{i=1}^{d}x_{i}=1\right\}
\]
be the probability simplex.

\begin{lem}[Doeblin contraction under uniform minorization]
\label{lem:positive_contraction}
Let $P$ be a $d\times d$ stochastic matrix satisfying $P_{i,j}\geq \alpha$ for all $i,j\in[d]$ for some $\alpha>0$.
Then, for all $\bm x,\bm y\in U$,
\[
\left\Vert \bm xP-\bm yP\right\Vert _{1}\leq(1-d\alpha)\left\Vert \bm x-\bm y\right\Vert_{1}.
\]
\end{lem}
\begin{proof}
For $\bm x,\bm y\in U$, observe that
\begin{align*}
\left\Vert \bm xP-\bm yP\right\Vert _{1}
&=\sum_{j=1}^{d}\left|\sum_{i=1}^{d}(x_{i}-y_{i})(P_{i,j}-\alpha+\alpha)\right|\\
&=\sum_{j=1}^{d}\left|\sum_{i=1}^{d}(x_{i}-y_{i})(P_{i,j}-\alpha)+\sum_{i=1}^{d}(x_{i}-y_{i})\alpha\right|\\
&=\sum_{j=1}^{d}\left|\sum_{i=1}^{d}(x_{i}-y_{i})(P_{i,j}-\alpha)\right|\\
&\leq\sum_{j=1}^{d}\sum_{i=1}^{d}\left|x_{i}-y_{i}\right|(P_{i,j}-\alpha)\\
&=\sum_{i=1}^{d}\left|x_{i}-y_{i}\right|\sum_{j=1}^{d}(P_{i,j}-\alpha)\\
&\leq\sum_{i=1}^{d}\left|x_{i}-y_{i}\right|(1-d\alpha)\\
&=(1-d\alpha)\left\Vert \bm x-\bm y\right\Vert _{1},
\end{align*}
where we used $\sum_{i=1}^{d}(x_i-y_i)=0$ and $P_{i,j}-\alpha\ge 0$.
\end{proof}

\begin{thm}[Exponential mixing from primitivity]
\label{thm:exp_mixing}
Let $Q$ be the transition matrix of a primitive Markov chain on $d$ states. Then there exist $k\in \mathbb{N}$ and $\alpha >0$ such that $[Q^k]_{i,j} \geq \alpha$ for all $i,j\in[d]$. For any such $k,\alpha$, we have
\begin{align*}
\|\bm{x} Q^t - \pi \|_1 \;\; & \leq 2 (1-d \alpha)^{\lfloor t/ k \rfloor} \\
\|\bm{x} Q^t - \pi \|_{\mathrm{TV}} &\leq (1-d \alpha)^{\lfloor t/ k \rfloor}
\end{align*}
for all probability row vectors $\bm{x}\in U$ and all integers $t\ge 0$, where $\pi$ is the stationary distribution.
\end{thm}
\begin{proof}
Since $Q$ is primitive, there exists $k$ such that $Q^k$ has strictly positive entries.
Let $\alpha= \min_{i,j}[Q^k]_{i,j}>0$ and apply Lemma~\ref{lem:positive_contraction} with $P=Q^k$ and $\bm y=\pi$ to get
\[
\| \bm{x} Q^k - \pi Q^k \|_1 = \| \bm{x} Q^k - \pi \|_1 \leq (1-d\alpha)\|\bm{x} - \pi \|_1.
\]
Since $\|\bm{x} - \pi \|_1 \leq 2$ for any $\bm x\in U$, iterating yields
\[
\|\bm{x} Q^{\ell k} - \pi \|_1 \le 2(1-d\alpha)^{\ell},\qquad \ell\in\mathbb{N}.
\]
For a general $t\ge 0$, write $t=\ell k + r$ with $\ell=\lfloor t/k\rfloor$ and $r=t\bmod k$. Then
\[
\|\bm{x} Q^{t} - \pi \|_1
=
\|(\bm{x}Q^{r}) Q^{\ell k} - \pi \|_1
\le 2(1-d\alpha)^{\ell}.
\]
This proves the stated inequality and the TV bound follows immediately.
\end{proof}

\begin{lem}[TV bound for decimated Markov chain]
\label{lem:spectral_bound}
Consider a primitive Markov chain on $d$ states with transition probability matrix $Q$ and state sequence $S_0,S_1,\ldots,S_{(n-1)T}$.
Fix $k\in \mathbb{N}$ and $\alpha>0$ such that $[Q^k]_{i,j}\ge \alpha$ for all $i,j\in[d]$.
For $T\ge k$ and $n\ge 1$, consider the decimated state $\bS^{\downarrow T} = (S_0,S_T,\ldots,S_{(n-1)T})$ which evolves according to the Markov transition matrix $Q^T$.

Let $P_{n,i}$ be the law for length-$n$ state sequences that start with $i$ copies of the stationary distribution and then evolve according to the decimated Markov chain.
Then, the decimated state sequence has distribution $P_{n,1}$ and the product of stationary marginals is given by $P_{n,n}$.
In this case, the total variation distance between these two distribution satisfies
\[
\|P_{n,1} - P_{n,n}\|_{\mathrm{TV}}
\ \le\
(n-1)\,(1-d\alpha)^{\lfloor T/k\rfloor}.
\]
\end{lem}

\begin{proof}
The decimated state sequence $\bS^{\downarrow T}$ %
is Markov with one-step kernel $Q^T$ and its distribution can be written as
\begin{align*}
K_m (s_1,\ldots,s_m | s_0) &\coloneqq \Pr \left( S^{\downarrow T}_1 \!= s_1, \ldots, S^{\downarrow T}_m \!= s_m \mid S^{\downarrow T}_0 \! = s_0 \right) \\
&= \prod_{i=0}^{m-1} [Q^T]_{s_i, s_{i+1}}.
\end{align*}
Now, we will compare this to the i.i.d.\ reference distribution
\[ \widetilde{K}_m (s_0,\ldots,s_{m-1}) \coloneqq \prod_{i=0}^{m-1} \pi(s_i). \] 
To do this, we define $P_{n,i}$ formally (for $1\le i\le n$) with
\begin{align*}
P_{n,i}&(\,s_0,\ldots,s_{n-1}) \\
&= \widetilde{K}_{i} (s_0,\ldots,s_{i-1}) K_{n-i} (s_i,\ldots, s_{n-1}|s_{i-1}).
\end{align*}
Thus, $P_{n,n}$ is the state distribution given by the product of stationary marginals and $P_{n,1}$ is the distribution of the decimated state sequence.

Using the triangle inequality, we get the following telescoping bound
\[
\|P_{n,1} - P_{n,n}\|_{\mathrm{TV}}
 \le \sum_{i=1}^{n-1} \|P_{n,i} - P_{n,i+1}\|_{\mathrm{TV}}.
\]
To simplify this, we calculate
\begin{align*}
&(P_{n,i}-P_{n,i+1}) (s_0,\ldots,s_{n-1}) = \widetilde{K}_{i} (s_0,\ldots,s_{i-1}) \\ & \;\; \cdot \left(\pi(s_i) K_{n-i-1} (s_{i+1},\ldots, s_{n-1}|s_{i}) - K_{n-i} (s_i,\ldots, s_{n-1}|s_{i-1})\right).
\end{align*}
After factoring out the common term $\widetilde{K}_i (s_0,\ldots,s_{i-1})$, we can upper bound the TV distance in the telescoping bound by averaging the conditional TV distance given $S_{i-1}=s_{i-1}$.
Since a Markov kernel cannot increase the TV distance, we find that
\begin{align*}
\|P_{n,i}-P_{n,i+1}\|_{\mathrm{TV}}
&\le
\sum_{s_{i-1} \in[d]} \pi(s_{i-1} )\,\|\pi  - e_{s_{i-1}} Q^T\|_{\mathrm{TV}} \\
&\le
(1-d\alpha)^{\lfloor T/k\rfloor},
\end{align*}
where the last step follows from Theorem~\ref{thm:exp_mixing} with $\mu=e_s$.
Using this in the telescoping sum, we get the stated result.
\end{proof}

\end{document}